\documentclass[12pt,a4paper]{article}
\usepackage[cmex10]{amsmath}
\usepackage{amssymb}
\usepackage[noend]{algpseudocode} 

\usepackage{array}

\usepackage{eqparbox}

 	\usepackage[pdftex]{graphicx, color}
  	\DeclareGraphicsExtensions{.pdf,.png,.jpg,.eps}
	\graphicspath{{figures/}}

\usepackage{epsfig}
\usepackage{graphicx}
\usepackage{epstopdf}
\DeclareGraphicsExtensions{.pdf,.png,.jpg,.eps}
\graphicspath{{figures/}}
\makeatletter
\def\input@path{{./figures/}}
\makeatother

\usepackage{url}
\usepackage{hyperref}
\usepackage{breakurl}

\makeatletter
\def\input@path{{./figures/}}
\makeatother

\newtheorem{theorem}{Theorem}[section]
\newtheorem{lemma}[theorem]{Lemma}

\newcommand{\sq}{\hbox{\rlap{$\sqcap$}$\sqcup$}}
\newcommand{\qed}{\hspace*{\fill}\sq}
\newenvironment{proof}{\noindent {\bf Proof.}\ }{\qed\par\vskip 4mm\par}

\newcommand{\remove}[1]{{}}

\algdef{SE}[DOWHILE]{Do}{doWhile}{\algorithmicdo}[1]{\algorithmicwhile\ #1}

\algnewcommand\EMPTY{\textbf{EMPTY}}

\algblockdefx[NAME]{Start}{End}%
	[1]{\textbf{Struct} #1:}%
	{EndStruct}
\algnotext[NAME]{End}

\algblockdefx[Name]{START}{END}%
	[1]{#1}%
	{Ending}
\algnotext[Name]{END}

\begin{document}
%
\title{DeltaTree: A Practical Locality-aware Concurrent Search Tree \\
(\it IFI-UIT Technical Report \bf2013-74)}

\date{October 10, 2013}


\author{Ibrahim Umar
\qquad Otto J. Anshus 
\qquad Phuong H. Ha 
\\ \texttt{\{ibrahim.umar, otto.anshus, phuong.hoai.ha\}@uit.no}
\\
\\ \small High Performance and Distributed Systems Group
\\ \small Department of Computer Science
\\ \small UiT The Arctic University of Norway
\\ \small (formerly University of Troms{\o})}

\maketitle

\begin{abstract}
As other fundamental programming abstractions in energy-efficient computing,
search trees are expected to support both high parallelism and data locality.
However, existing highly-concurrent search trees such as red-black trees and AVL
trees, do not consider data locality while existing locality-aware search trees
such as those based on the van Emde Boas layout (vEB-based trees), poorly
support concurrent (update) operations.

This paper presents DeltaTree, a practical locality-aware concurrent search tree
that combines both locality-optimisation techniques from vEB-based trees and
con-currency-optimisation techniques from non-blocking highly-concurrent search
trees.  DeltaTree is a $k$-ary leaf-oriented tree of DeltaNodes in which each
DeltaNode is a size-fixed tree-container with the van Emde Boas layout.
The expected memory transfer costs of DeltaTree's {\em Search, Insert} and {\em
Delete} operations are $O(log_B N)$, where $N, B$ are the tree size and the {\em
unknown} memory block size in the ideal cache model, respectively.
 DeltaTree's {\em Search} operation is wait-free, providing prioritised lanes for {\em
 Search} operations, the dominant operation in search trees. Its {\em Insert}
 and {\em Delete} operations are non-blocking to other {\em Search, Insert} and
 {\em Delete} operations, but they may be occasionally blocked by maintenance
 operations that are sometimes triggered to keep DeltaTree in good shape. Our
 experimental evaluation using the latest implementation of AVL, red-black, and speculation
 friendly trees from the Synchrobench benchmark has shown that DeltaTree is up
 to 5 times faster than all of the three concurrent search trees for searching
 operations and up to 1.6 times faster for update operations when the update
 contention is not too high.
\end{abstract}

\clearpage
\tableofcontents
\clearpage
\listoffigures
\clearpage 
\listoftables 
\clearpage

\section{Introduction}

Energy efficiency is becoming a major design constraint in current and future
computing systems ranging from embedded to high performance computing (HPC)
systems.  In order to construct energy efficient software systems, data
structures and algorithms must support not only high parallelism but also data
locality \cite{Dally11}.
Unlike conventional locality-aware data structures and algorithms that concern
only whether data is on-chip (e.g. data in cache) or not (e.g. data in DRAM),
new energy-efficient data structures and algorithms must consider data locality
in finer-granularity: where on chip the data is. It is because in modern
manycore systems the energy difference between accessing data in nearby
memories (2pJ) and accessing data across the chip (150pJ) is almost two orders
of magnitude, while the energy difference between accessing on-chip data (150pJ)
and accessing off-chip data (300pJ) is only two-fold \cite{Dally11}. Therefore,
fundamental data structures and algorithms such as search trees need to support
both high parallelism and fine-grained data locality.

However, existing highly-concurrent search trees do not consider fine-grained
data locality. The highly concurrent search trees includes non-blocking
\cite{EllenFRB10, Brown:2011:NKS:2183536.2183551} and Software Transactional
Memory (STM) based search trees
\cite{Afek:2012:CPC:2427873.2427875, BronsonCCO10,
Crain:2012:SBS:2145816.2145837, DiceSS2006}. The prominent highly-concurrent search trees included in several benchmark distributions are the concurrent red-black
tree \cite{DiceSS2006} developed by Oracle Labs and the  concurrent AVL tree
developed by Stanford \cite{BronsonCCO10}. The highly concurrent trees,
however, do not consider the tree layout in memory for data locality.

Concurrent B-trees \cite{BraginskyP12, Comer79, Graefe:2010:SBL:1806907.1806908,
Graefe:2011:MBT:2185841.2185842}
are optimised for a known memory block size $B$ (e.g. page size) to minimise the
number of memory blocks accessed during a search, thereby improving data
locality. 
As there are different block sizes at different levels of  the memory hierarchy
(e.g.
register size, SIMD width,  cache line size and page size) that can be utilised
to design locality-aware layout for search trees \cite{KimCSSNKLBD10},
concurrent B-trees limits its spatial locality optimisation to the memory level
with block size $B$, leaving memory accesses to the other memory levels
unoptimised.
For example, if the concurrent B-trees are optimised for accessing disks (i.e.
$B$ is the page size), the cost of searching a key in a block of size $B$ in
memory is $\Theta (log (B/L))$ cache line transfers, where $L$ is the cache line
size \cite{BrodalFJ02}. Since each memory read basically contains only one node
of size $L$ from a top down traversal of a path in the search tree of $B/L$
nodes, except for the topmost $\lfloor log(L+1) \rfloor$ levels. Note that the
optimal cache line transfers in this case is $O(log_L B)$, which is achievable
by using the van Emde Boas layout.

A van Emde Boas (vEB) tree is an ordered dictionary data type which implements
the idea of recursive structure of priority
queues \cite{vanEmdeBoas:1975:POF:1382429.1382477}.
The efficient structure of the vEB tree, especially how it arranges data
in a recursive manner so that related values are placed in contiguous memory
locations, has inspired cache oblivious (CO) data structures \cite{Prokop99} 
such as CO B-trees \cite{BenderDF05, BenderFGK05, BenderFFFKN07} and CO
binary trees \cite{BrodalFJ02}. These researches have demonstrated that the
locality-aware structure of the vEB layout is a perfect fit for cache oblivious algorithms,
lowering the upper bound on memory transfer complexity.

\begin{figure}[!t]
\centering \scalebox{0.8}{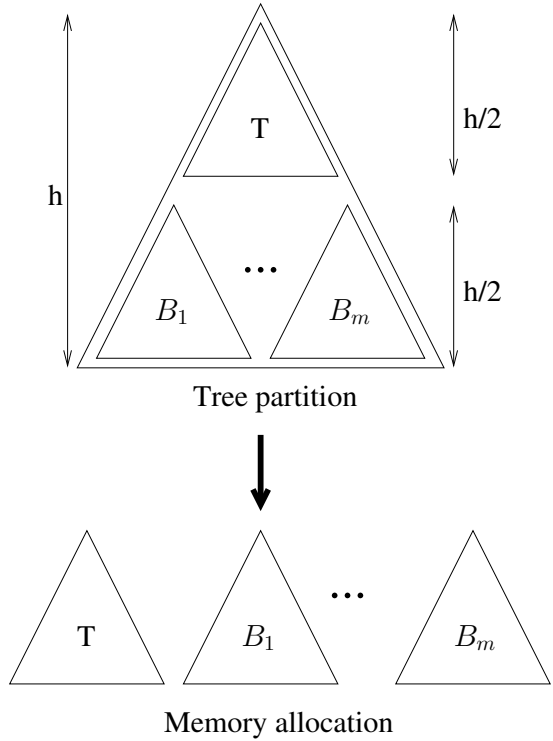}
\caption{An illustration for the van Emde Boas layout}\label{fig:vEB}
\end{figure}

Figure \ref{fig:vEB} illustrates the vEB layout. A tree of height $h$ is
conceptually split between nodes of heights $h/2$ and $h/2 + 1$, resulting in a
top subtree $T$ of height $h/2$ and  $m = 2^{h/2}$ bottom subtrees $B_1, B_2,
\cdots, B_m$ of height $h/2$. The $(m+1)$ subtrees are located in contiguous
memory locations in the order $T, B_1, B_2, \cdots, B_m$. Each of the subtrees of
height $h/2^i, i \in \mathbb{N},$ is recursively partitioned into $(m+1)$
subtrees of height $h/2^{i+1}$ in a similar manner, where $m=2^{h/2^{i+1}}$,
until each subtree contains only one node. With the vEB layout, the search cost is $O(log_B N)$ memory
transfers, where $N$ is the tree size and $B$ is the {\em unknown} memory block
size in the I/O \cite{AggarwalV88} or
ideal-cache\cite{Frigo:1999:CA:795665.796479} model.
The search cost is optimal and matches the search bound of B-trees that requires
the memory block size $B$ to be known in advance. More details on the vEB layout
are presented in Section \ref{sec:dynamic_veb}.

The vEB-based trees, however, poorly support concurrent update operations.
Inserting or deleting a node in a tree may result in relocating a large part of
the tree in order to maintain the vEB layout. For example, inserting a node in
full subtree $T$ in Figure \ref{fig:vEB} will affect the other subtrees $B_1,
B_2, \cdots, B_m$  due to shifting them to the right in the memory, or even
allocating a new contiguous block of memory for the whole tree, in order to have
space for the new node \cite{BrodalFJ02}. Note that the subtrees $T, B_1, B_2,
\cdots, B_m$ must be located in {\em contiguous} memory locations according to
the vEB layout. The work in  \cite{BenderFGK05} has discussed the problem but
not yet come out with a feasible implementation \cite{BraginskyP12}.

We introduce $\Delta$Tree, a novel locality-aware concurrent search
tree that combines both locality-optimisation techniques from vEB-based trees
and concurrency-optimisation techniques from non-blocking highly-concurrent
search trees. Our contributions are threefold:

\begin{itemize}
  \item We introduce a new {\em relaxed} cache oblivious model and a novel {\em
  dynamic} vEB layout that makes the vEB layout suitable for highly-concurrent 
  data structures with update operations. The dynamic vEB layout supports
  dynamic node allocation via pointers while maintaining the optimal search 
  cost of $O(log_B N)$ memory transfers without knowing the exact value of $B$
  (cf. Lemma \ref{lem:dynamic_vEB_search}). 
  The new relaxed cache-oblivious model and dynamic vEB layout are presented in
  Section \ref{sec:dynamic_veb}.
   
  \item Based on the new dynamic vEB layout, we develop $\Delta$Tree, a novel
  locality-aware concurrent search tree. $\Delta$Tree is a $k$-ary leaf-oriented
  tree of $\Delta$Nodes in which each $\Delta$Node is a size-fixed
  tree-container with the van Emde Boas layout. The expected memory transfer
  costs of $\Delta$Tree's {\em Search, Insert} and {\em Delete} operations are
  $O(log_B N)$, where $N$ is the tree size and $B$ is the {\em unknown} memory
  block size in the ideal cache model \cite{Frigo:1999:CA:795665.796479}.
  $\Delta$Tree's {\em Search} operation is wait-free while its {\em Insert} and
  {\em Delete} operations are non-blocking to other {\em Search, Insert} and
  {\em Delete} operations, but they may be occasionally blocked by maintenance
  operations. $\Delta$Tree overview is presented in Section \ref{sec:overview}
  and its detailed implementation and analysis are presented in Section
  \ref{sec:implementation}.
  
  \item We experimentally evaluate $\Delta$Tree on commodity machines, comparing
  it with the prominent concurrent search trees such as AVL trees
  \cite{BronsonCCO10}, red-black trees \cite{DiceSS2006} and speculation
  friendly trees \cite{Crain:2012:SBS:2145816.2145837} from the Synchrobench
  benchmark \cite{synchrobench}. The experimental results show that $\Delta$Tree
  is up to 5 times faster than all of the three concurrent search trees for
  searching operations and up to 1.6 times faster for update operations when the
  update contention is not too high. We have also developed a concurrent version
  of the sequential vEB-based tree in \cite{BrodalFJ02} using GCC's STM in order
  to gain insights into the performance characteristics of concurrent vEB-based
  trees. The detailed experimental evaluation is presented in Section
  \ref{sec:evaluation}. 
 The code of the $\Delta$Tree and its experimental evaluation are 
 available upon request. 
\end{itemize}

\section{Dynamic Van Emde Boas Layout}\label{sec:dynamic_veb}

\subsection{Notations}

We first define these notations that will be used hereafter in this paper:
\begin{itemize}

\item $b_i$ (unknown): block size in term of nodes at level $i$ of memory
hierarchy (like $B$ in the I/O model
\cite{AggarwalV88}), which is unknown as in the cache-oblivious
model \cite{Frigo:1999:CA:795665.796479}. When the specific level $i$ of memory
hierarchy is irrelevant, we use notation $B$ instead of $b_i$ in order to be
consistent with the I/O model.

\item $UB$ (known): the upper bound (in terms of the number of nodes) on the
block size $b_i$ of all levels $i$ of the memory hierarchy.

\item {\em $\Delta$Node}: the coarsest recursive subtree of a vEB-based search
tree that contains at most $UB$ nodes (cf. dash triangles of height $2^L$ in
Figure \ref{fig:search_complexity}). $\Delta$Node is a size-fixed tree-container
with the vEB layout.

\item Let $L$ be the level of detail of $\Delta$Nodes. Let $H$ be the height
of a $\Delta$Node, we have $H = 2^L$. For simplicity, we assume $H = log_2
(UB+1)$.

\item $N, T$: size and height of the whole tree in terms of basic nodes (not in
terms of $\Delta$Nodes).

\item $density(r) = n_r / UB$ is the density of $\Delta$Node rooted at
$r$,  where $n_r$ the number of nodes currently stored in the $\Delta$Node.

\end{itemize}

\subsection{Static van Emde Boas (vEB) Layout}

The conventional {\em static} van Emde Boas (vEB) layout has been introduced in
cache-oblivious data structures \cite{BenderDF05, BrodalFJ02,
Frigo:1999:CA:795665.796479}. Figure \ref{fig:vEB} illustrates the vEB layout.
Suppose we have a complete binary tree with height $h$. For simplicity, we
assume $h$ is a power of 2, i.e. $h=2^k$.
The tree is recursively laid out in the memory as follows. The tree is
conceptually split between nodes of height $h/2$ and $h/2+1$, resulting in a top
subtree $T$ and $m_1 = 2^{h/2}$ bottom subtrees $B_1, B_2, \cdots, B_m$ of
height $h/2$.  The $(m_1 +1)$ top and bottom subtrees are then located in
consecutive memory locations in the order of subtrees $T, B_1, B_2, \cdots,
B_m$. Each of the subtrees of height $h/2$ is then laid out similarly to $(m_2 +
1)$ subtrees of height $h/4$, where $m_2 = 2^{h/4}$. The process continues until
each subtree contains only one node, i.e. the finest {\em level of detail}, 0.
Level of detail $d$ is a partition of the tree into recursive subtrees of
height at most $2^d$.

The main feature of the vEB layout is that the cost of any search in this layout
is $O(log_B N)$ memory transfers, where $N$ is the tree size and $B$ is the {\em
unknown} memory block size in the I/O \cite{AggarwalV88} or 
ideal-cache \cite{Frigo:1999:CA:795665.796479} model. The search cost is the
optimal and matches the search bound of B-trees that requires the memory block
size $B$ to be known in advance. Moreover, at any level of detail, each subtree
in the vEB layout is stored in a contiguous block of memory.

Although the vEB layout is helpful for utilising data locality, it poorly
supports concurrent update operations. Inserting (or deleting) a node at
position $i$ in the contiguous block storing the tree may restructure a large
part of the tree stored after node $i$ in the memory block. For example, inserting
new nodes in the full subtree $A$ in Figure \ref{fig:vEB} will  affect the other
subtrees $B_1, B_2, \cdots, B_m$ by shifting them to the right in order to have
space for new nodes. Even worse, we will need to allocate a new contiguous block of
memory for the whole tree if the previously allocated block of memory for
the tree runs out of space \cite{BrodalFJ02}. Note that we cannot use dynamic
node allocation via pointers since at {\em any} level of detail, each subtree in the vEB layout must be stored in a {\em
contiguous} block of memory.

\subsection{Relaxed Cache-oblivious Model and Dynamic vEB Layout}

\begin{figure}[!t] 
\centering \scalebox{0.8}{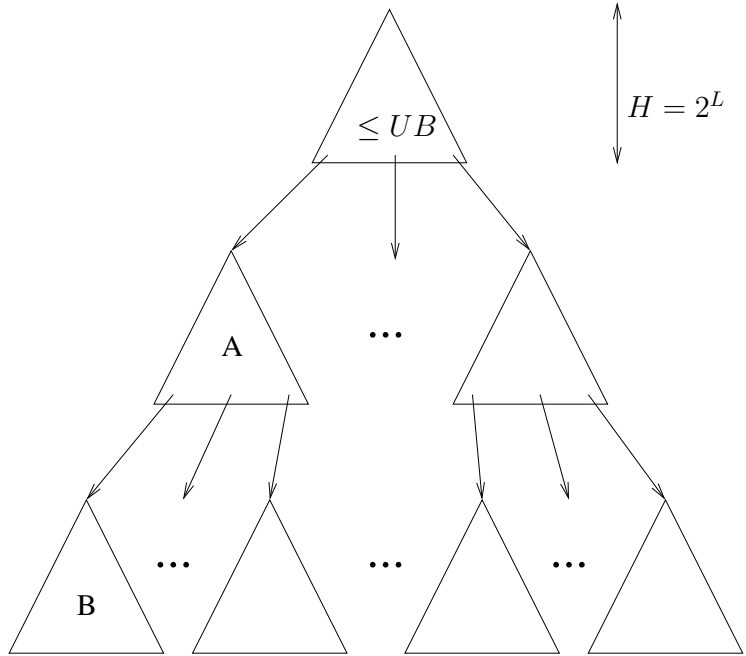}
\caption{An illustration for the new dynamic vEB layout}\label{fig:dynamicVEB}
\end{figure}

In order to make the vEB layout suitable for highly concurrent data structures
with update operations, we introduce a novel {\em dynamic} vEB layout. Our key
idea is that if we know an upper bound $UB$ on the unknown memory block size
$B$, we can support dynamic node allocation via pointers while maintaining the
optimal search cost of $O(log_B N)$ memory transfers without knowing $B$ (cf.
Lemma \ref{lem:dynamic_vEB_search}). 

We define {\em relaxed cache oblivious} algorithms to be cache-oblivious (CO)
algorithms with the restriction that an upper bound $UB$ on the unknown memory
block size $B$ is known in advance. As long as an upper bound on all the block
sizes of multilevel memory is known, the new relaxed CO model maintains the key
feature of the original CO model, namely analysis for a simple two-level memory
are applicable for an unknown multilevel memory (e.g. registers, L1/L2/L3 caches
and memory). This feature enables designing algorithms that can utilise
fine-grained data locality in energy-efficient chips \cite{Dally11}. In
practice, although the exact block size at each level of the memory hierarchy is
architecture-dependent (e.g. register size, cache line size), obtaining a single upper
bound for all the block sizes (e.g. register size, cache line size and page size)
is easy. For example, the page size obtained from the operating system is such
an upper bound.
        
Figure \ref{fig:dynamicVEB} illustrates the new dynamic vEB layout based on the
relaxed cache oblivious model. Let $L$ be the coarsest level of detail such that
every recursive subtree contains at most $UB$  nodes. The tree is recursively
partitioned into level of detail $L$ where each subtree represented by a
triangle in Figure \ref{fig:dynamicVEB},  is stored in a contiguous memory block
of size $UB$. Unlike the conventional vEB, the subtrees at level of detail $L$
are linked to each other using pointer, namely each subtree at level of detail
$k > L$ is not stored in a contiguous block of memory.  Intuitively, since $UB$
is an upper bound on the unknown memory block size $B$, storing a subtree at
level of detail $k > L$ in a contiguous memory block of size greater than $UB$,
does not reduce the number of memory transfer. For example, in Figure
\ref{fig:dynamicVEB}, a travel from a subtree $A$ at level of detail $L$, which
is stored in a contiguous memory block of size $UB$, to its child subtree $B$ at
the same level of detail will result in at least two memory transfers: one for
$A$ and one for $B$. Therefore, it is unnecessary to store both $A$ and $B$ in a
contiguous memory block of size $2UB$. As a result, the cost of any search in
the new dynamic vEB layout is intuitively the same as that of the conventional
vEB layout (cf. Lemma \ref{lem:dynamic_vEB_search}) while the former supports
highly concurrent update operations because it utilises pointers.

Let $\Delta$Node be a subtree at level of detail $L$, which is stored in a
contiguous memory block of size $UB$ and is represented by a triangle in Figure
\ref{fig:dynamicVEB}.

\begin{figure}[!t] 
\centering \scalebox{0.8}{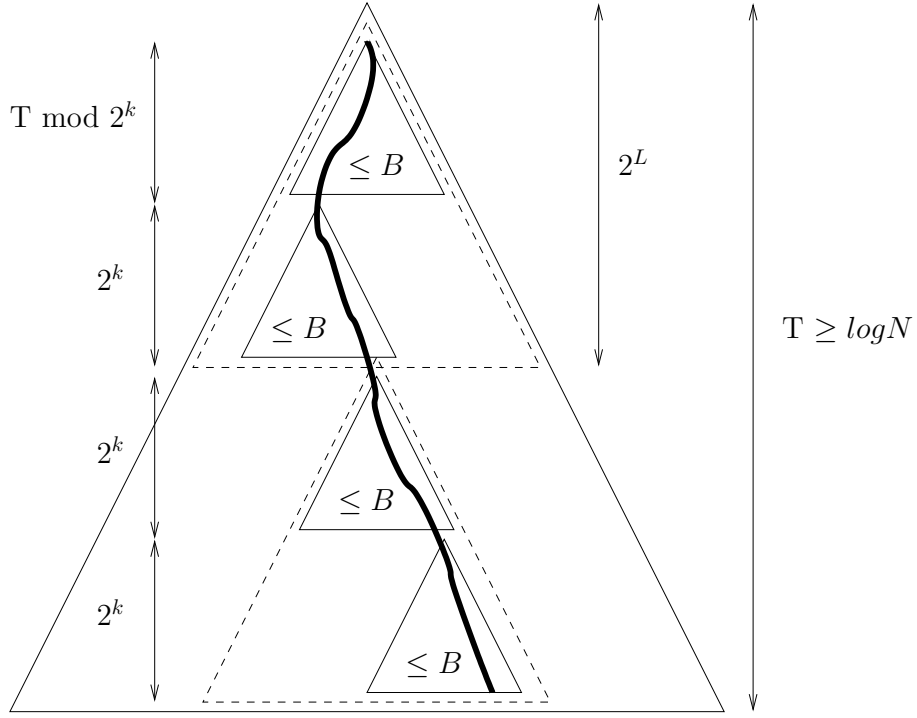}
\caption{Search illustration}\label{fig:search_complexity}
\end{figure}

\begin{lemma}
A search in a complete binary tree with the new dynamic vEB layout needs 
$O(log_B N)$ memory transfers, where $N$ and $B$ is the tree size and 
the {\em unknown} memory block size in the ideal cache model
\cite{Frigo:1999:CA:795665.796479}, respectively.
\label{lem:search_mem} \label{lem:dynamic_vEB_search}
\end{lemma}
\begin{proof} (Sketch)
Figure \ref{fig:search_complexity} illustrates the proof.  
Let $k$ be the coarsest level of detail such that every recursive subtree 
contains at most $B$ nodes. Since $B \leq UB$, $k \leq L$, where $L$ is 
the coarsest level of detail at which every recursive subtree 
contains at most $UB$ nodes. That means there are at most $2^{L-k}$ subtrees
along to the search path in a $\Delta$Node and no subtree of depth $2^k$ is split 
due to the boundary of $\Delta$Nodes. Namely, triangles of height $2^k$ fit
within a dash triangle of height $2^L$ in Figure \ref{fig:search_complexity}. 

Due to the property of the new dynamic vEB layout that at any level of detail 
$i \leq L$, a recursive subtree of depth $2^i$ is stored in a contiguous block 
of memory, each subtree of depth $2^k$ {\em within} a $\Delta$Node is stored in
at most 2 memory blocks of size $B$ (depending on the starting location of
the subtree in memory). Since every subtree of depth $2^k$ fits in a
$\Delta$Node (i.e.
no subtree is stored across two $\Delta$Nodes), every subtree of depth $2^k$ is 
stored in at most 2 memory blocks of size $B$.

Since the tree has height $T$, $\lceil T / 2^k \rceil$ subtrees of depth $2^k$ 
are traversed in a search and thereby at most  $2  \lceil T / 2^k \rceil$ memory 
blocks are transferred. 

Since a subtree of height $2^{k+1}$ contains more than $B$ nodes, 
$2^{k+1} \geq log_2 (B + 1)$, or $2^{k} \geq \frac{1}{2} log_2 (B+ 1)$. 

We have $2^{T-1} \leq N \leq 2^T$ since the tree is a {\em complete} binary tree. 
This implies $ log_2 N \leq T \leq log_2 N +1$.  

Therefore, $2  \lceil T / 2^k \rceil \leq 4 \lceil \frac{log_2 N + 1}{log_2 (B + 1)} \rceil 
= 4 \lceil log_{B+1} N + log_{B+1} 2\rceil = O(log_B N)$, where $N \geq 2$.
\end{proof}

\section{$\Delta$Tree Overview}\label{sec:overview}

Figure \ref{fig:treeuniverse} illustrates a $\Delta$Tree named $U$. $\Delta$Tree
$U$ uses our new dynamic vEB layout presented in Section \ref{sec:dynamic_veb}.
The $\Delta$Tree consists of $|U|$ $\Delta$Nodes of fixed size $UB$ each of
which contains a \textit{leaf-oriented} binary search tree (BST) $T_i, i=1,
\dots,|U|$. $\Delta$Node's internal nodes are put together in cache-oblivious
fashion using the vEB layout.

\begin{figure}[!t] \centering \includegraphics[width=0.8\columnwidth ]{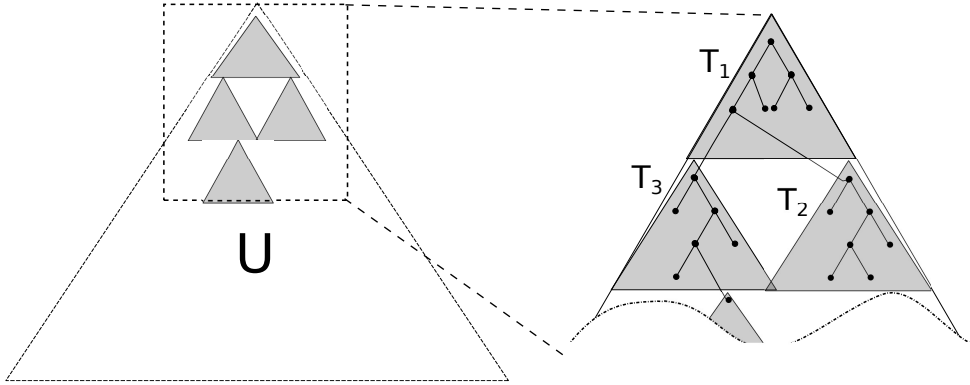}
\caption{Depiction of $\Delta$Tree universe $U$}
\label{fig:treeuniverse}
\end{figure}

The $\Delta$Tree $U$ acts as the dictionary of abstract data types. It maintains
a set of values which are members of an ordered universe \cite{EllenFRB10}. It offers the
following operations: \textsc{insertNode($v, U$)}, which adds value $v$ to the
set $U$, \textsc{deleteNode($v, U$)} for removing a value $v$ from the set, and
\textsc{searchNode($v, U$)}, which determines whether value $v$ exists in the
set. We may use the term \textit{update} operation for either insert or delete
operation. We assume that duplicate values are not allowed inside the set and a
special value, say $0$, is reserved as an indicator of an \textsc{Empty}
value.

Operation \textsc{searchNode($v,U$)} is going to walk over the $\Delta$Tree to
find whether the value $v$ exists in $U$. This particular operation is
guaranteed to be wait-free, and returning \textbf{true} whenever $v$ has been
found, or \textbf{false} otherwise. The \textsc{insertNode($v, U$)} inserts a
value $v$ at the leaf of $\Delta$Tree, provided $v$ does not yet exist in the
tree. Following the nature of a leaf-oriented tree, a successful insert operation 
will replace a leaf with a subtree of three nodes \cite{EllenFRB10} (cf. Figure
\ref{fig:treetransform}a).
The \textsc{deleteNode($v, U$)} simply just {\em marks} the leaf that contains
the value $v$ as deleted, instead of physically removing the leaf or changing
its parent pointer as proposed in \cite{EllenFRB10}.

Apart from the basic operations, three maintenance $\Delta$Tree operations are
invoked in certain cases of inserting and deleting a node from the tree.
Operation \textsc{rebalance($T_v.root$)} is responsible for rebalancing a
$\Delta$Node after an insertion.
Figure \ref{fig:treetransform}a illustrates the rebalance operation. Whenever a
new node $v$ is to be inserted at the last level $H$ of $\Delta$Node $T_1$, the
$\Delta$Node is rebalanced to a complete BST by setting a new depth for all
leaves (e.g. $a,v,x,z$ in Figure \ref{fig:treetransform}a) to $\log N + 1$,
where $N$ is the number of leaves. In Figure \ref{fig:treetransform}a, we can
see that after the rebalance operation, tree $T_1$ becomes more balanced and its
height is reduced from 4 into 3.

We also define the \textsc{expand($v$)} operation, that will be responsible for
creating new $\Delta$Node and connecting it to the parent of the leaf node $v$.
Expand will be triggered only if after \textsc{insertNode($v, U$)}, the leaf $v$
will be at the last level of a $\Delta$Node and rebalancing will no longer
reduce the current height of the subtree $T_i$ stored in the $\Delta$Node. 
For example if the expanding is taking place at a node $v$ located at the bottom
level of the $\Delta$Node (Figure \ref{fig:treetransform}b), or $depth(v) = H$,
then a new $\Delta$Node ($T_2$ for example) will be referred by the parent of
node $v$, immediately after value of node $v$ is copied to $T_2.root$ node.
Namely, the parent of $v$ swaps one of its pointer that previously points to
$v$, into the root of the newly created $\Delta$Node, $T_2.root$.

The \textsc{merge($T_v.root$)} is for merging $T_v$ with its sibling 
after a node deletion. For example, in Figure \ref{fig:treetransform}c $T_2$ is
merged into $T_3$. Then
the pointer of $T_3$'s grandparent that previously points
to the parent of both $T_3$ and $T_2$ is replaced to point to $T_3$.
The operations are invoked provided that a particular $\Delta$Node where 
the deletion takes place, is filled less than half of its capacity and all
values of that $\Delta$Node and its siblings can be fitted into a $\Delta$Node. 

To minimise block transfers required during tree traversal, the height of the
tree should be kept minimal. These auxiliary operations are the unique
feature of $\Delta$Tree in the effort of maintaining a small height.

These \textsc{insertNode} and \textsc{deleteNode} operations are non-blocking 
to other  \textsc{searchNode, insertNode} and \textsc{deleteNode} operations.
Both of the operations are using single word CAS (Compare and Swap)
and "leaf-checking" to achieve that. Section \ref{sec:implementation}
will explain more about these update operations.

As a countermeasure against unnecessary waiting for concurrent maintenance
operations, a buffer array is provided in each of the $\Delta$Nodes. This buffer
has a length that is equal to the number of maximum concurrent threads. As an
illustration of how it works, consider two concurrent
operations \textsc{insertNode($v,U$)} are operating inside the same
$\Delta$Node.
Both are successful and have determined that expanding or rebalancing are
necessary. Instead of rebalancing twice, those two threads will then compete to obtain the lock on that $\Delta$Node. The losing
thread will just append $v$ into the buffer and then exits. The winning thread, 
which has successfully acquired the lock, will do rebalancing or expanding using
all the leaves and the buffer of that $\Delta$Node. The same process happens for
concurrent delete, or the mix of concurrent insert and delete.

Despite \textsc{insertNode} and \textsc{deleteNode} are non-blocking, they still
need to wait at the tip of a $\Delta$Node whenever either of the maintenance
operations, \textsc{rebalance} and \textsc{merge} is currently operating within
that $\Delta$Node. We employ TAS (Test and Set) using $\Delta$Node lock to make
sure that no basic update operations will interfere with the maintenance
operations. Note that the previous description has shown that \textsc{rebalance}
and \textsc{merge} execution are actually sequential within a $\Delta$Node, so
reducing the invocations of those operations is crucial to deliver a scalable
performance of the update operations. To do this, we have set a density
threshold that acts as limiting factor, bringing a good amortised cost of
insertion and deletion within a $\Delta$Node, and subsequently for the whole
$\Delta$Tree. 
The proof for the amortised cost are given in Section \ref{sec:implementation}
of this paper.

Concerning the $\textsc{expand}$ operation, an amount of memory for a new
$\Delta$Node needs to be allocated during runtime. Since we kept the size of a
$\Delta$Node equal to the page size, memory allocation routine for new
$\Delta$Nodes does not require plenty of CPU cycles.

\begin{figure}[!t] \centering \includegraphics[width=0.8\columnwidth]{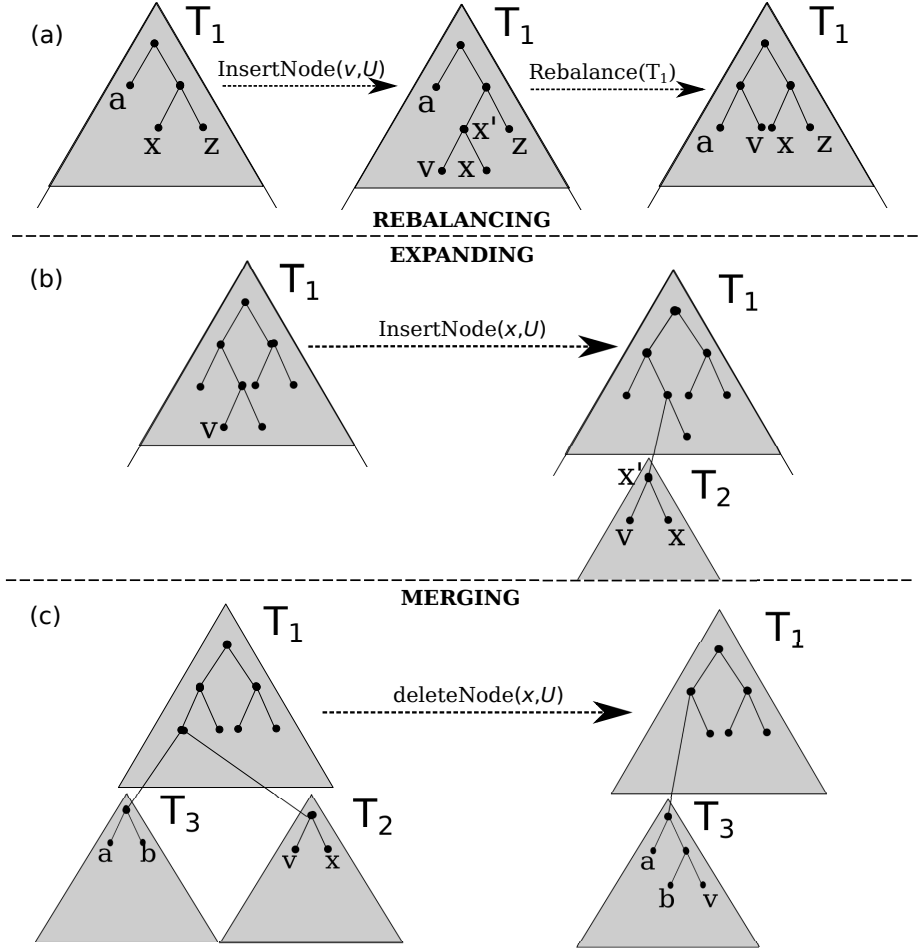}
\caption{(a)\textit{Rebalancing}, (b)\textit{Expanding}, and (c)\textit{Merging} operations on $\Delta$Tree}
\label{fig:treetransform}
\end{figure}


\section{Detailed Implementation}\label{sec:implementation}


\subsection{Function specifications}
The \textsc{searchNode}($v, U$) function will return \textbf{true} whenever value $v$
has been inserted in one of the $\Delta$Tree ($U$) leaf node and that node's $mark$
property is currently set to \textbf{false}. Or if $v$ is placed on one of the 
$\Delta$Node's buffer located at the lowest level of $U$. It returns \textbf{false} whenever
it couldn't find a leaf node with $value=v$, or $v$ couldn't be found in 
the last level $T_{tid}.rootbuffer$. 

\textsc{insertNode}($v, U$) will insert value $v$ and returns \textbf{true} 
if there is no leaf node with $value=v$, or
there is a leaf node $x$ which satisfy $x.value=v$ but with $x.mark=\textbf{true}$, 
or $v$ is not found in the last $T_{tid}$'s $rootbuffer$.
In the other hand, \textsc{insertNode} returns \textbf{false} if there is a leaf node with 
$value=v$ and $mark=\textbf{false}$, or $v$ is found in $T_{tid}.rootbuffer$.

For \textsc{deleteNode}($v, U$), a value of \textbf{true} is returned if there
is a leaf node with $value=v$ and $mark=\textbf{false}$, or $v$ is found in the
last $T_{tid}$'s $rootbuffer$. The value $v$ will be then deleted.
In the other hand, \textsc{deleteNode} returns \textbf{false} if there is a leaf
node with $value=v$ and $mark=\textbf{true}$, or $v$ is not found in
$T_{tid}.rootbuffer$.

\subsection{Synchronisation calls}
For synchronisation between update and maintenance operations, we define
\textsc{flagup}($opcount$) that is doing atomic increment of $opcount$ and also
a function that do atomic decrement of $opcount$ as
\textsc{flagdown}($opcount$).

\begin{figure}[!t] \centering 
\begin{algorithmic}[1] 
\small
\Function{waitandcheck}{$lock, opcount$}
    \Do
        \State \textsc{spinwait}($lock$) \label{lst:line:spinwait}
        \State \textsc{flagup}($opcount$) 
        \State $repeat \gets \textbf{false}$
        \If{$lock=true$}
            \State \textsc{flagdown}($opcount$) 
            \State $repeat \gets \textbf{true}$
        \EndIf
    \doWhile{$repeat=\textsc{true}$}
\EndFunction
\end{algorithmic}
\caption{Wait and check algorithm}
\label{lst:waitandcheck}
\end{figure}

Also there is \textsc{spinwait}($lock$) that basically instruct a thread to
spin waiting while $lock$ value is \textbf{true}. Only \textsc{Merge} and 
\textsc{Rebalance} that will have to privilege to set $T_x.lock$ as \textbf{true}.
Lastly there is
\textsc{waitandcheck}($lock, opcount$) function (Figure
\ref{lst:waitandcheck}) that is preventing updates in getting mixed-up with maintenance
operations. The mechanism of \textsc{waitandcheck}($lock, opcount$)
will instruct a thread to wait at the tip of a current $\Delta$Node
whenever another thread has obtained a lock on that $\Delta$Node
for the purpose of doing any maintenance operations.

\begin{figure}[!t] \centering \begin{algorithmic}[1]
\small
\Start{\textbf{node} $n$}
	\START{member fields:}
\State $tid \in \mathbb{N}$, if $> 0$ indicates the node is root of a \par
\hskip\algorithmicindent $\Delta$Node with an id of $tid$ ($T_{tid}$) 

\State $value \in
\mathbb{N}$, the node value, default is \textbf{empty} 

\State $mark \in
\{true,false\}$, a value of $true$ indicates a logically \par \hskip\algorithmicindent deleted
node \State $left, right \in \mathbb{N}$, left / right child pointers \State $isleaf
\in {true,false}$, indicates whether the \par \hskip\algorithmicindent node is a
leaf of a $\Delta$Node, default is $true$ 	\label{lst:line:leafdefault}

\END \End

\Statex
\Start{\textbf{$\Delta$Node} $T$}
	\START{member fields:}
\State $nodes$, a group of  $(|T| \times 2)$ amount of\par
\hskip\algorithmicindent pre-allocated node $n$.
\State $rootbuffer$, an array of value with a length \par
\hskip\algorithmicindent  of the current number of threads \State
$mirrorbuffer$, an array of value with a length \par \hskip\algorithmicindent 
of the current number of threads \State $lock$, indicates whether a $\Delta$Node is
locked \State $flag$, semaphore for active update operations \State $root$,
pointer the root node of the $\Delta$Node \State $mirror$, pointer to root node of
the $\Delta$Node's \par \hskip\algorithmicindent mirror \END \End

\Statex
\Start{\textbf{universe} $U$}
	\START{member fields:}
\State $root$, pointer to the $root$ of the topmost $\Delta$Node \par
\hskip\algorithmicindent ($T_1.root$) \END \End
\end{algorithmic}
\caption{Cache friendly binary search tree structure} \label{lst:datastruct}
\end{figure}

\begin{figure}[!t] \centering \begin{algorithmic}[1]
\small
\Function{searchNode}{$v, U$}
\State $lastnode, p \gets U.root$
\While{$p \neq$ not end of tree $\And p.isleaf \neq$ TRUE}	\label{lst:line:searchifleaf}
    \State $lastnode \gets p$	\label{lst:line:lasnode-p}
        \If{$p.value < v$}		\label{lst:line:searchless}
            \State $p \gets p.left$
        \Else					\label{lst:line:searchelse}
            \State $p \gets p.right$
        \EndIf
\EndWhile
\If{$lastnode.value = v$} \label{lst:line:linsearch3}
	\If{$lastnode.mark =$ FALSE}				\label{lst:line:linsearch1}
		\State \Return TRUE 
	\Else
		\State \Return FALSE
	\EndIf
\Else \State Search ($T_{tid}.rootbuffer$) for $v$	\label{lst:line:searchbuffer}
	\If{$v$ is found}							\label{lst:line:linsearch2}
\State\Return TRUE \Else \State\Return FALSE \EndIf \EndIf \EndFunction
\end{algorithmic}
\caption{A wait-free searching algorithm of $\Delta$Tree}\label{lst:nodeSearch}
\end{figure}

\begin{figure}
\centering
\begin{algorithmic}[1]
\small

\Function{insertNode}{$v,U$} \Comment{\parbox[t]{.5\linewidth}{Inserting an new item $v$ into $\Delta$Tree $U$}}
\State $t \gets U.root$
\State \Return \textsc{insertHelper}($v,t$)
\EndFunction
\State

\Function{deleteNode}{$v,T$} \Comment{\parbox[t]{.5\linewidth}{Deleting an item $v$ from $\Delta$Tree $U$}}
\State $t \gets U.root$
\State \Return \textsc{deleteHelper}($v,t$)
\EndFunction
\State

\Function{deleteHelper}{$v, node$}						
\State $success \gets$ TRUE
\If{Entering new $\Delta$Node $T_x$}             \Comment{\parbox[t]{.5\linewidth}{Observed by examining $x \gets node.tid$ value change}}    			                                                 
   \State  $T'_x \gets$ \textsc{getParent$\Delta$Node}($T_x$)
   \State  \textsc{flagdown}($T'_x.opcount$) 	\Comment{\parbox[t]{.5\linewidth}{Flagging down operation count on the previous/parent triangle}}			
   \State  \textsc{waitandcheck}($T_x.lock$, $T_x.opcount$)
   \State  \textsc{flagup}($T_x.opcount$)
\EndIf

\If{($node.isleaf =$ TRUE)} 			\Comment{Are we at leaf?}
        \If{$node.value = v$}
            \If{\textsc{CAS}($node.mark$, FALSE, TRUE) != FALSE)}              	\label{lst:line:markdel}       \Comment{Mark it delete}
                	\State $success \gets$ FALSE                                           		 \Comment{Unable to mark, already deleted}
            \Else	
             	\If{($node.left.value$=\textbf{empty}\&$node.right.value$=\textbf{empty})} 	\label{lst:line:markdel-check} 
			\State $T_x.bcount \gets T_x.bcount - 1 $	
                		\State \textsc{mergeNode}(\textsc{parentOf}($T_x)) \gets$ TRUE      \Comment{Delete succeed, invoke merging}                 
		\Else		 
			\State \textsc{deleteHelper}($v, node$)	 		\Comment{\parbox[t]{.5\linewidth}{Not leaf, re-try delete from $node$}}
		\EndIf
            \EndIf
        \Else												
        	\State Search ($T_{x}.rootbuffer$) for $v$
	   	\If{$v$ is found in $T_{x}.rootbuffer.idx$}
                 		\State $T_{x}.rootbuffer.idx \gets \textbf{empty}$	\label{lst:line:bufdel} 
				\State $T_x.bcount\gets T_x.bcount - 1$
        			\State $T_x.countnode\gets T_x.countnode - 1$
            	\Else	
			\State \textsc{flagdown}($T_x.opcount$)
            		\State $success \gets$ FALSE					\Comment{\parbox[t]{.5\linewidth}{Value not found}}
	    	\EndIf                                             
    	\EndIf
	\State \textsc{flagdown}($T_x.opcount$)
\Else \If{$v < node.value$}
        \State \textsc{deleteHelper}($v, node.left$)
\Else
        \State \textsc{deleteHelper}($v, node.right$)
\EndIf
\EndIf
\State \Return $success$
\EndFunction

\algstore{myalg}
\end{algorithmic}
\end{figure}

\begin{figure}
\centering
\begin{algorithmic} [1]                   
\algrestore{myalg}
\small

\Function{insertHelper}{$v, node$}
\State $success \gets$ TRUE
\If{Entering new $\Delta$Node $T_x$}       				\Comment{\parbox[t]{.5\linewidth}{Observed by examining $x \gets node.tid$ change}}                                     	                      
   \State  $T'_x \gets$ \textsc{getParent$\Delta$Node}($T_x$)	 
   \State  \textsc{flagdown}($T'_x.opcount$) 				\Comment{\parbox[t]{.5\linewidth}{Flagging down operation count on the previous/parent triangle}}
   \State  \textsc{waitandcheck}($T_x.lock$, $T_x.opcount$)
	\State  \textsc{flagup}($T_x.opcount$)
\EndIf

\If{$node.left \And node.right$}                          \Comment{\parbox[t]{.5\linewidth}{At the lowest level of a $\Delta$Tree?}}                          
        \If{$v < node.value$}                                                                              
            \If{($node.isleaf =$ TRUE)}                                 
                \If{CAS($node.left.value$, \textbf{empty}, $v$) = \textbf{empty}}       \label{lst:line:ins-lp1} 	\label{lst:line:growins1} 
                	\State $node.right.value \gets node.value$                                                       
                    \State $node.right.mark \gets node.mark$
                    \State $node.isleaf \gets$ FALSE	
                    \State  \textsc{flagdown}($T_x.opcount$)
                \Else
                    \State \textsc{insertHelper}($v$, $node$)     \Comment{\parbox[t]{.5\linewidth}{Else try again to insert starting with the same node}}                     
                \EndIf
            \Else
                \State \textsc{insertHelper}($v$, $node.left$)             \Comment{Not a leaf, proceed further to find the leaf}        
            \EndIf
        \ElsIf {$v > node.value$}                                                                      
            \If{($node.isleaf =$ TRUE)}                                  
                \If{CAS($node.left.value$, $\textbf{empty}$, $v$) =
                \textbf{empty}}  \label{lst:line:ins-lp2}  \label{lst:line:growins2}
                	\State $node.right.value \gets v$                                                        
                    \State $node.left.mark \gets node.mark$
                    \State \textsc{atomic} $\{$					\label{lst:line:atomic-isleaf}
                    	$node.value \gets v$
                    	\State $node.isleaf \gets$ FALSE
                    $\}$	
                  \State  \textsc{flagdown}($T_x.opcount$)
                \Else
                    \State \textsc{insertHelper}($v$, $node$)                      \Comment{\parbox[t]{.5\linewidth}{Else try again to insert starting with the same node}}     
                \EndIf
            \Else
                \State \textsc{insertHelper}($v$, $node.right$)                	   \Comment{Not a leaf, proceed further to find the leaf}
            \EndIf
        \ElsIf {$v = node.value$}                                                                                            
            \If{($node.isleaf =$ TRUE)}                                   
                \If{$node.mark =$ FALSE}				
                    \State $success \gets$ FALSE            				\Comment{\parbox[t]{.5\linewidth}{Duplicate Found}}
                   \State  \textsc{flagdown}($T_x.opcount$)                                   
                \Else
                	   \State Goto \ref{lst:line:growins2} 	
                \EndIf
            \Else
                \State \textsc{insertHelper}($v, node.right$)                    	 \Comment{Not a leaf, proceed further to find the leaf}      
            \EndIf
        \EndIf
\Else			
	\If{$val = node.value$}                                                                            
            \If{$node.mark = 1$}
                \State $success \gets FALSE$                                               
            \EndIf
        \Else 					\Comment{\parbox[t]{.5\linewidth}{All's failed, need to rebalance or expand the triangle $T_x$}}
\algstore{myalg}
\end{algorithmic}
\end{figure}

\begin{figure}[!t]
\centering
\begin{algorithmic} [1]                   
\algrestore{myalg}
\small
		\If{$v$ already in $T_x.rootbuffer$}
			$success \gets FALSE$   
		\Else
			\State \textit{put $v$ inside $T_x.rootbuffer$} \label{lst:line:insertbuffer}
			\State $T_x.bcount\gets T_x.bcount + 1$
        			\State $T_x.countnode\gets T_x.countnode + 1$
		\EndIf            
          	\If{\textsc{TAS}($T_x.lock$)}                                            			\Comment{\parbox[t]{.5\linewidth}{All threads try to lock $T_x$}}

                		\State \textsc{flagdown}($T_x.opcount$)                                      \Comment{\parbox[t]{.5\linewidth}{Make sure no flag is still raised}}
                		\State \textsc{spinwait}($T_x.opcount$)                                           \Comment{\parbox[t]{.5\linewidth}{Now wait all insert/delete operations to finish}}
                
                		\State $total \gets T_x.countnode + T_x.bcount$
                		
                		\If {$total * 4 \geqslant U.maxnode + 1$}                             \Comment{Expanding needed, $density > 0.5$}
                    		\State \ldots \textit{Create(a new triangle) AND attach it on the to the parent of $node$} \ldots                   
                		\Else
                    		\If{$T_x$ don't have triangle child(s)}
                        			\State $T_x.mirror \gets$ \textsc{rebalance}($T_x.root$, $T_x.rootbuffer$)
                        			\State \textsc{switchtree}($T_x.root$, $T_x.mirror$)
                        			\State $T_x.b_count \gets 0$
                    		\Else
                        			\If{$T_x.bcount > 0$}
                            			\State Fill $childA$ with all value in $T_x.rootbuffer$ 	\Comment{\parbox[t]{.2\linewidth}{Do non-blocking insert here}}
                            			\State $T_x.bcount \gets 0$
                        			\EndIf                             
				\EndIf
                		\EndIf	
			\State \textsc{spinunlock}($T_x.lock$)   
		\Else
			\State \textsc{flagdown}($T_x.opcount$) 
		\EndIf    
	\EndIf                                  
\EndIf
\State \Return $success$
\EndFunction
\end{algorithmic}
\caption{Update algorithms and their helpers functions}\label{lst:pseudo-ops}
\end{figure}

\begin{figure}[!t]
\begin{algorithmic}[1]
\small
\Procedure{balanceTree}{$T$}
\State Array $temp[|H|]\gets$ Traverse($T$) \Comment{\parbox[t]{.5\linewidth}{Traverse all the non-empty node into $temp$ array}}
\State RePopulate($T, temp$) \Comment{\parbox[t]{.5\linewidth}{Re-populate the tree $T$ with all the value from $temp$ recursively. RePopulate will resulting a balanced tree $T$}}
\EndProcedure

\Procedure{mergeTree}{$root$}
\State $parent \gets$ \textsc{parentOf}($root$)

\If{$parent.left = root$}
	\State $sibling \gets parent.right$  
\Else
	\State $sibling \gets parent.left$
\EndIf

\State $T_r \gets$ \textsc{triangleOf}($root$)			\Comment{\parbox[t]{.5\linewidth}{Get the $T$ of $root$ node}}
\State $T_p \gets$ \textsc{triangleOf}($parent$)			\Comment{\parbox[t]{.5\linewidth}{Get the $T$ of $parent$}}
\State $T_s \gets$ \textsc{triangleOf}($sibling$)			\Comment{\parbox[t]{.5\linewidth}{Get the $T$ of $sibling$}}

\If{\textsc{spintrylock}($T_r.lock$)}                \Comment{\parbox[t]{.5\linewidth}{Try to lock the current triangle}}
      	\State \textsc{spinlock}($T_s.lock, T_p.lock$)	 \Comment{\parbox[t]{.5\linewidth}{lock the sibling triangles}}
	\State \textsc{flagdown}($T_r.opcount$) 
	\State \textsc{spinwait}($T_r.opcount, T_s.opcount, T_p.opcount$)	 \Comment{\parbox[t]{.3\linewidth}{Wait for all insert/delete operations to finish}}

	\State $total \gets T_s.nodecount + T_s.bcount + T_r.nodecount + T_r.bcount$

	\If{($T_s \And T_r$ don't have children) $\And (T_p \geqslant U.maxnode + 1) / 2) \| T_s \geqslant U.maxnode + 1) / 2)) \And total \leqslant (U.maxnode + 1) / 2$}
		\State MERGE $T_r.root , T_r.rootbuffer, T_s.rootbuffer$ into $T.s$ 
		\If{$parent.left = root$}				\Comment{\parbox[t]{.5\linewidth}{Now re-do the pointer}}
			\State $parent.left \gets root.left$  	\Comment{\parbox[t]{.5\linewidth}{Merge Left}}
		\Else
			\State $parent.right \gets root.right$	\Comment{\parbox[t]{.5\linewidth}{Merge Right}}
		\EndIf
	\EndIf
   	\State \textsc{spinunlock}($T_r.lock, T_s.lock, T_p.lock$)
\Else
	\State \textsc{flagdown}($T_r.opcount$) 
\EndIf

\EndProcedure
\end{algorithmic}
\caption{Merge and Balance algorithm}\label{mergeTree}
\end{figure}

\subsection{Wait-free and Linearisability of search}

\begin{lemma}
$\Delta$Tree search operation is wait-free.
\end{lemma}

\begin{proof}(Sketch) In the searching algorithm (cf. Figure
\ref{lst:nodeSearch}), the $\Delta$Tree will be traversed from the root node
using iterative steps. When at $root$, the value to search $v$ is compared to
$root.value$. If $v < root.value$, the left side of the tree will be traversed
by setting $root \gets root.left$ (line \ref{lst:line:searchless}), in contrary
$v \geqslant root.value$ will cause the right side of the tree to be traversed
further (line \ref{lst:line:searchelse}). The procedure will repeat until a leaf
has been found ($v.isleaf = \textbf{true}$) in line \ref{lst:line:searchifleaf}.

If the value $v$ couldn't be found and search has reached the end of 
$\Delta$Tree, a buffer search will be conducted
(line \ref{lst:line:searchbuffer}).
This search is done by simply searching the buffer array from left-to-right
to find $v$, therefore no waiting will happen in this phase.

The \textsc{deleteNode} and \textsc{insertNode} algorithms (Figure
\ref{lst:pseudo-ops}) are non-intrusive to the structure of a tree, thus they
won't interfere with an ongoing search. A \textsc{deleteNode} operation, if
succeeded, is only going to mark a node by setting a $v.mark$ variable as
\textbf{true} (line \ref{lst:line:markdel} in Figure \ref{lst:pseudo-ops}).
The $v.value$ is retained so that a search will be able to proceed further.
For \textsc{insertNode}, it can \textit{"grow"} the current leaf node as it
needs to lays down two new leaves (lines  \ref{lst:line:growins1} and
\ref{lst:line:growins2} in Figure \ref{lst:pseudo-ops}), however the operation
never changes the internal pointer structure of a $\Delta$Node, since
$\Delta$Node internal tree structure is pre-allocated beforehand, allowing a
search to keep moving forward.
As depicted in Figure \ref{fig:treetransform}(a), after an insertion of $v$
grows the node, the old node (now $x'$) still contains the same value as $x$
(assuming $v < x$), thus a search still can be directed to find either $v$ or
$x$.
 The \textsc{rebalance/Merge} operation is also not an obstacle for searching
 since its operating on a mirror $\Delta$Node.
\end{proof}

We have designed the searching to be linearisable in various concurrent
operation scenarios (Lemma \ref{lem:linear-search}). This applies as well to the
update operations.

\begin{lemma} 
For a value that resides on the leaf node of a $\Delta$Node,
\textsc{searchNode} operation (Figure \ref{lst:nodeSearch}) has the linearisation 
point to \textsc{deleteNode} at line \ref{lst:line:linsearch1} and the
linearisation point to \textsc{insertNode} at line \ref{lst:line:linsearch3}. 
For a value that stays in the buffer of a $\Delta$Node, \textsc{searchNode} 
operation has the \textit{linearisation point} at line \ref{lst:line:linsearch2}.
\label{lem:linear-search}
\end{lemma}
\begin{proof}(Sketch) It is trivial to demonstrate this in relation to
deletion algorithm in Figure \ref{lst:pseudo-ops} since only an atomic operation
is responsible for altering the $mark$ property of a node (line
\ref{lst:line:markdel}).
Therefore \textsc{deleteNode} has the linearisation point to \textsc{searchNode} 
at line \ref{lst:line:markdel}. 
 
For \textsc{searchNode} interaction with an insertion that grows new subtree, we
rely on the facts that: 1) a snapshot of the current node $p$ is recorded on
$lastnode$ as a first step of searching iteration (Figure \ref{lst:nodeSearch},
line \ref{lst:line:lasnode-p}); 2) $v.value$ change, if needed, is not done
until the last step of the insertion routine for insertion of $v > node.value$
and will be done in one atomic step with $node.isleaf$ change (Figure
\ref{lst:pseudo-ops}, line \ref{lst:line:atomic-isleaf}); and 3) $isleaf$
property of all internal nodes are by default \textbf{true} (Figure
\ref{lst:datastruct}, line \ref{lst:line:leafdefault}) to guarantee that values
that are inserted are always found, even when the leaf-growing (both
left-and-right) are happening concurrently. Therefore
\textsc{insertNode} has the linearisation point to \textsc{searchNode} at line 
\ref{lst:line:ins-lp1} when inserting a value $v$ smaller than the leaf node's 
value, or at line \ref{lst:line:ins-lp2} otherwise. 

A search procedure is also able to cope well with a "buffered" insert, that is
if an insert thread loses a competition in locking a $\Delta$Node for expanding
or rebalancing and had to dump its carried value inside a buffer (Figure
\ref{lst:pseudo-ops}, line \ref{lst:line:insertbuffer}). Any value inserted to
the buffer is guaranteed to be found. This is because after a leaf $lastnode$
has been located, the search is going to evaluate whether the $lastnode.value$
is equal to $v$. Failed comparison will cause the search to look further inside
a buffer ($T_x.rootbuffer$) located in a $\Delta$Node where the leaf resides
(Figure \ref{lst:nodeSearch}, line \ref{lst:line:searchbuffer}).
By assuring  that the switching of a root $\Delta$Node with its mirror includes
switching $T_x.rootbuffer$ with $T_x.mirrorbuffer$, we can show that any new
values that might be placed inside a buffer are guaranteed to be found
immediately after the completion of their respective insert procedures.
The "buffered" insert has the linearisation point to \textsc{searchNode} at
line \ref{lst:line:insertbuffer}.

Similarly, deleting a value from a buffer is as trivial as exchanging that value
inside a buffer with an \textbf{empty} value. The search operation will failed
to find that value when doing searching inside a buffer of $\Delta$Node. This
type of delete has the linearisation point to \textsc{searchNode} at the same
line it's emptying a value inside the buffer (line \ref{lst:line:bufdel}).
\end{proof}

\subsection{Non-blocking Update Operations} 

\begin{lemma}
$\Delta$Tree Insert and Delete operations are
non-blocking to each other in the absence of maintenance operations.
\end{lemma}

\begin{proof} (Sketch)
Non-blocking update operations supported by $\Delta$Tree are possible by
assuming that any of the updates are not invoking \textsc{Rebalance} and
\textsc{Merge} operations. 
In a case of concurrent insert operations (Figure \ref{lst:pseudo-ops}) 
at the
same leaf node $x$, assuming all insert threads need to "grow" the node 
(for illustration, cf. Figure \ref{fig:treetransform}), they
will have to do \textsc{CAS}($x.left, \textbf{empty},\ldots$) 
(line  \ref{lst:line:growins1} and \ref{lst:line:growins2}) as their
first step. This CAS is the only thing needed since the whole $\Delta$Node
structure is pre-allocated and the CAS is an atomic operation. Therefore, only
one thread will succeed in changing $x.left$ and proceed populating the
$x.right$ node. Other threads will fail the CAS operation and they are going to
try restart the insert procedure all over again, starting from the node $x$.

To assure that the marking delete (line \ref{lst:line:markdel}) behaves nicely
with the "grow" insert operations, \textsc{deleteNode}($v, U$) that has found
the leaf node $x$ with a value equal to $v$, will need to check again whether
the node is still a leaf (line \ref{lst:line:markdel-check}) after completing
\textsc{CAS}($x.mark, FALSE, TRUE$).
The thread needs to restart the delete process from $x$ if it has found that $x$
is no longer a leaf node.

The absence of maintenance operations means that a $\Delta$Node $lock$ is 
never set to \textbf{true}, thus either insert/delete operations are never blocked
at the execution of line number \ref{lst:line:growins2} in Figure \ref{lst:waitandcheck}.
\end{proof}

\begin{lemma} 
In Figure \ref{lst:pseudo-ops}, \textsc{insertNode} operation has the linearisation 
point against \textsc{deleteNode} at line \ref{lst:line:growins1} and line \ref{lst:line:growins2}. 
Whereas \textsc{deleteNode} has a linearisation point at line \ref{lst:line:markdel-check}
against an \textsc{insertNode} operation. 
For inserting and deleting into a buffer of a $\Delta$Node, an \textsc{insertNode} 
operation has the \textit{linearisation point} at line \ref{lst:line:insertbuffer}. While
\textsc{deleteNode} has its linearisation point at line \ref{lst:line:bufdel}.
\label{lem:linear-ins-del}
\end{lemma}
\begin{proof}(Sketch) An \textsc{insertNode} operation will do a \textsc{CAS} on the left node
as its first step after finding a suitable $node$ for growing a subtree. If value $v$
is lower than  $node.value$, the correspondent operation is the line \ref{lst:line:growins1}.
Line \ref{lst:line:growins2} is executed in other conditions. A \textsc{deleteNode} will always
check a $node$ is still a leaf by ensuring $node.left.value$ as \textbf{empty}
(line \ref{lst:line:markdel-check}). This is done after it tries to mark that $node$. If the comparison
on line \ref{lst:line:markdel-check} returns \textbf{true}, the operation finishes successfully.
A \textbf{false} value will instruct the \textsc{insertNode} to retry again, starting from
the current node.

A buffered insert and delete are operating on the same buffer. When a value $v$ is put inside
a buffer it will always available for delete. And that goes the opposite for the deletion case.  
\end{proof}

\subsection{Memory Transfer and Time Complexities} 
 
In this subsection, we will show that $\Delta$Tree is relaxed cache oblivious
and the overhead of maintenance operations (e.g. rebalancing, expanding and
merging) is negligible for big trees. The memory transfer analysis is based on
the ideal-cache model \cite{Frigo:1999:CA:795665.796479}. Namely, re-accessing
data in cache due to re-trying in non-blocking approaches incurs no memory
transfer.

For the following analysis, we assume that values to be searched, inserted or
deleted are randomly chosen. As $\Delta$Tree is a binary search tree (BST),
which is embedded in the dynamic vEB layout, the expected height of a randomly
built $\Delta$Tree of size $N$ is $O(\log N)$ \cite{CormenSRL01}.

\begin{lemma}
A search in a randomly built $\Delta$Tree needs $O(\log_B N)$ expected memory
transfers, where $N$ and $B$ is the tree size and the {\em unknown} memory
block size in the ideal cache model \cite{Frigo:1999:CA:795665.796479}, 
respectively. 
\label{lem:DeltaTree_search}
\end{lemma}
\begin{proof}(Sketch)
Similar to the proof of Lemma \ref{lem:dynamic_vEB_search}, let $k, L$ be the
coarsest levels of detail such that every recursive subtree contains at most
$B$ nodes or $UB$ nodes, respectively. Since $B \leq UB$, $k \leq L$. There are
at most $2^{L-k}$ subtrees along to the search path in a $\Delta$Node and no
subtree of depth $2^k$ is split due to the boundary of $\Delta$Nodes (cf. Figure
\ref{fig:search_complexity}). Since every subtree of depth $2^k$ fits in a
$\Delta$Node of size $UB$, the subtree is stored in at most 2 memory blocks of
size $B$.

Since a subtree of height $2^{k+1}$ contains more than $B$ nodes, 
$2^{k+1} \geq \log_2 (B + 1)$, or $2^{k} \geq \frac{1}{2} \log_2 (B+ 1)$. 

Since a randomly built $\Delta$Tree has an expected height of $O(\log N)$, there
are $\frac{O(\log N)}{2^k}$ subtrees of depth $2^k$ are traversed in a search and
thereby at most  $2  \frac{O(\log N)}{2^k} = O(\frac{\log N}{2^k})$ memory
blocks are transferred.

As $\frac{\log N}{2^k} \leq 2\frac{\log N}{\log (B+ 1)} = 2 \log_{B+1} N \leq 2
log_B N$, expected memory transfers in a search are $O(\log_B N)$.
\end{proof}

\begin{lemma}
Insert and Delete operations within the $\Delta$Tree are having a similar 
amortised time complexity of $O(\log n + UB)$, where $n$ is the size of $\Delta$Tree,
and $UB$ is the maximum size of element stored in $\Delta$Node.
\label{lem:amortised_cost}
\end{lemma}

\begin{proof} (Sketch) An insertion operation at $\Delta$Tree is tightly coupled
with the rebalancing and expanding algorithm. 

We assume that $\Delta$Tree was built using random values, therefore the 
expected height is $\mathcal{O}(\log n)$. 
Thus, an insertion on a $\Delta$Tree costs $\mathcal{O}(\log n)$. 
Rebalancing after insertion only happens
at single $\Delta$Node, and it has an upper bound cost of $\mathcal{O}(UB + UB \log UB)$,
because it has to read all the stored elements, sort it out and re-insert it in a balanced fashion.  
In the worst possible case for $\Delta$Tree, there will be an $n$ insertion 
that cost $\log n$ and there is at most $n$ rebalancing operations with a cost of
$\mathcal{O}(UB + UB \log UB)$ each. 

Using aggregate analysis, we let total cost for insertion 
as $\displaystyle\sum_{k=1}^{n} c_{i} \leqslant n\log n + \displaystyle\sum_{k=1}^{n} UB + UB \log UB
\approx n \log n + n \cdot (UB + UB \log UB)$. Therefore the amortised time complexity 
for insert is $\mathcal{O}(\log n + UB + UB \log UB)$. If we have defined $UB$ such as 
$UB<<n$, the amortised time complexity for inserting a value into $\Delta$Tree 
is now becoming $\mathcal{O}(\log n)$.

For the expanding scenarios, an insertion will trigger \textsc{expand}($v$)
whenever an insertion of $v$ in a $\Delta$Node $T_j$ is resulting on $depth(v) = H(T_j)$ and
$|T_j|\geqslant (2^{H(T_j)-1})-1$. An expanding will require a memory allocation
of a $UB$-sized $\Delta$Node, cost merely $\mathcal{O}(1)$, together with two pointer
alterations that cost $\mathcal{O}(1)$ each. In conclusion, we have shown that
the total amortised cost for insertion, that is incorporating both rebalancing
and expanding procedures as $\mathcal{O}(\log n)$.

In the deletion case, right after a deletion on a particular $\Delta$Node 
will trigger a merging of that $\Delta$Node with its sibling in a condition 
of at least one of the $\Delta$Nodes is filled less than half
of its maximum capacity ($density(v) < 0.5$) and all values from both 
$\Delta$Nodes can fit into a single $\Delta$Node.

Similar to insertion, a deletion in $\Delta$Tree costs $\log n$. However merging that combines
2 $\Delta$Nodes costs $2UB$ at maximum. Using aggregate analysis, 
the total cost of deletion could be formulated as
$\displaystyle\sum_{k=1}^{n} c_{i} \leqslant n\log n + \displaystyle\sum_{k=1}^{n} 2\cdot UB
\approx n \log n + 2n \cdot UB$. The amortised time complexity is 
therefore $\mathcal{O}(\log n + UB)$ or $\mathcal{O}(\log n)$, if $UB<<n$.

\end{proof}


\section{Experimental Result and Discussion}\label{sec:evaluation}


To evaluate our conceptual idea of $\Delta$Tree, we compare its implementation 
performance with
those of STM-based AVL tree (AVLtree), red-black tree (RBtree), and Speculation
Friendly tree (SFtree) in the Synchrobench benchmark
\cite{synchrobench}. We also have developed an STM-based 
binary search tree which is based on the work of \cite{BrodalFJ02} utilising 
GNU C Compiler's STM implementation from the version 4.7.2 .
This particular tree will be referred as VTMtree, and it has all the traits of vEB tree
layout, although it only has a fixed size, which is pre-defined before the
runtime. Pthreads were used for concurrent threads 
and the GCC were invoked with -O2 optimisation to
compile all of the programs. 

The base of the conducted experiment consists of running a series
of ($rep=100,000,000$) operations. Assuming we have $nr$ as the number of
threads, the time for a thread to finish a sequence of $rep/nr$ operations will
be recorded and summed with the similar measurement from the other threads. We
also define an update rate $u$ that translates to $upd = u\% \times rep$ number of
insert and delete operations and $src = rep - upd$ number of search operations
out of $rep$. We set a consecutive run for the experiment to use a combination of
update rate $u=\{0, 1, 3, 5, 10, 20, 100\}$ and number of thread 
$nr=\{1, 2, \ldots, 16\}$ for each runs.
Update rate of 0 means that only searching operations were conducted, while 100
update rate indicates that no searching were carried out, only insert and delete
operations. For each of the combination above, we pre-filled the initial tree using 
1,023 and 2,500,000 values.
A $\Delta$Tree with initial members of 1,023 increases the chances that a thread
will compete for a same resources and also simulates a condition where the whole tree
fits into the  cache. The initial size of 2,500,000 lowers the chance of thread contentions 
and simulates a big tree that not all of it will fits into the last level of cache.
The operations involved (e.g. searching, inserting or deleting) used random
values $v \in (0, 5,000,000], v \in \mathbb{N}$, as their parameter 
for searching, inserting or deleting. Note that VTMtree is fixed in size, 
therefore we set its size to 5,000,000 to accommodate this experiment.

The conducted experiment was run on a dual Intel Xeon CPU E5-2670, for a total of
16 available cores. The node had 32GB of memory, with a 2MB L2 cache and 
a 20MB L3 cache. The Hyperthread feature of the processor was turned off to 
make sure that the experiments only ran on physical cores.
The performance (in operations/second) result for update operations were calculated 
by adding the number successful insert and delete. While searching performance were
using  the number of attempted searches. Both were divided by the total time
required to finish $rep$ operations. 

In order to satisfy the locality-aware properties of the $\Delta$Tree, we need 
to make sure that the size of $\Delta$Nodes, or $UB$, 
not only for Lemma \ref{lem:dynamic_vEB_search} to hold true, 
but also to make sure that all level of the memory hierarchy (L1, L2, ... caches) are
efficiently utilised, while also minimising the frequency of false sharing in a
highly contended concurrent operation. For this we have tested various value for
$UB$, using $127$, $1K-1$, $4K-1$, and $512K-1$ sized elements, and by assuming a node size
in the $\Delta$Node is 32 bytes. These values will correspond to $4$ Kbytes (page
size for most systems), L1 size, L2 size, and L3 size respectively.
Please note that L1, L2, and L3 sizes here are measured in our test system. Based on
the result of this test, we found out that $UB = 127$ delivers the best
performance results, in both searching and updating. This is in-line with the
facts that the page size is the block size used during memory allocation
 \cite{Smith:1982:CM:356887.356892, Drepper07whatevery}. This improves 
the transfer rate from main memory to the CPU cache. 
Having a $\Delta$Node that
fits in a page will help the $\Delta$Tree in exploiting the data locality
property.

\begin{figure}\centering
\resizebox{0.65\columnwidth}{!}{\input{perf-1023}}
\caption{Performance rate (operations/second)  of a $\Delta$Tree with 1,023 initial members. The y-axis indicates the rate of operations/second.}
\label{fig:perf-1023}
\end{figure}

As shown in Figure \ref{fig:perf-1023}, under a small tree setup, 
the $\Delta$Tree has a better update
performance (i.e. insertion and deletion) compared to the other trees, whenever
the update ratio is less than 10\%. From the said figure, 10\% update ratio
seems to be the cut-off point for $\Delta$Tree before SFtree, AVLtree, and
RBtree gradually took over the performance lead. Even though the update rate of
the $\Delta$Tree were severely hampered after going on higher than 10\% update
ratio, it does manage to keep a comparable performance for a small number of threads. 

For the search performance evaluation using the same setup, $\Delta$Tree
is superior compared to other types of tree when the search ratio higher than
90\% (cf. Figure \ref{fig:perf-1023}). In the extreme case of 100\% search ratio
(i.e. no update operation), $\Delta$Tree does however get beaten by the
VTMtree.

\begin{figure}\centering
\resizebox{0.65\columnwidth}{!}{\input{perf-2500000}}
\caption{Performance rate (operations/second) of a $\Delta$Tree with 2,500,000 initial members. The y-axis indicates the rate of operations/second.}
\label{fig:perf-2500000}
\end{figure}

On the other setup, the big tree setup with an initial member of 2,500,000 nodes
(cf. Figure \ref{fig:perf-2500000}), a slightly different result on update performance can be
observed. Here the $\Delta$Tree maintains a lead in the concurrent update
performance up to 20\% update ratio. Higher ratio than this diminishes the
$\Delta$Tree concurrent update performance superiority. Similar to what can bee
seen at the small tree setup, during the extreme case of 100\% update ratio
(i.e. no search operation), the $\Delta$Tree seems to be able to kept its pace
for 6 threads, before flattening-out in the long run, losing out to the SFtree,
AVLtree, and RBtree.  VTMtree update performance is the worst.

As for the concurrent searching performance in the same setup, the $\Delta$Tree
outperforms the other trees when the search ratio is
less than 100\%. At the 80 \% search ratio, the VTMtree search performance is the
worst and the search performance of the other four trees is comparable. At the extreme case of
100\% search ratio, the VTMtree performance is the best.

The $\Delta$Tree performs well in the low-contention situations. Whenever
the a big tree setup is used, the $\Delta$Tree delivers scalable updating and
searching performance up to 20\% update ratio, compared to only 10\% update
ratio in the small tree setup.
The good update performance of $\Delta$Tree can be attributed to the dynamic vEB
layout that permits that multiple different $\Delta$Nodes can be concurrently updated and
restructured. Keeping the frequency of restructuring done by \textsc{Merge} and 
\textsc{Rebalance} at low also contribute to this good performance. 
In terms of searching, the $\Delta$Tree have been showing an overall good
performance, which only gets beaten by the static vEB layout-based VTMtree at
the extreme case of 100\% searching ratio.

In order to get better insight into the performance $\Delta$Tree, we conducted additional 
experiment targeting the cache behaviour of the different trees. In this experiment, two flavours
of $\Delta$Tree, one using $\Delta$Node size of 127 and another using a size of
5,000,000, together with both VTMtree and SFtree were put to do 100M searching
operations. Big $\Delta$Node size in this experiment simulates a {\em leaf-oriented} static
vEB, with only 1 $\Delta$Node involved, whereas the VTMtree simulates a {\em
original} static vEB where values can be stored at internal nodes.
Those trees are pre-filled with 1,048,576 random non-recurring numbers
within $(0,5,000,000]$ range. The values searched for were randomly picked as
well within the same range. Cache profiles were then collected using Valgrind \cite{valgrind}
Our test system has 20MB of CPU's L3 cache,
therefore the pre-initialised nodes were not entirely contained within the cache
($1048576 \times 32$B $> 20$MB). This experiment result in Table
\ref{tbl:cache} proved that using the dynamic vEB layout were indeed able to reduce
the number of cache misses by almost 2\%. This is observed by comparing the 
percentage of cache misses between leaf-oriented static vEB $\Delta$Tree ($UB=5M$) 
and leaf-oriented dynamic vEB $\Delta$Tree ($UB=127$). 
However it doesn't translate to a higher update rate due to increasing load count. 

It is interesting to see that VTMtree is able to deliver the lowest load count as well as
the lowest number of cache misses. This result leads us to conclude that using
leaf-oriented tree for the sake of supporting scalable concurrent updates, has a
downside of introducing more cache misses. This can be related to the fact that a search
in leaf-oriented tree has to always traverse all the way down to the leaves.
Although using dynamic vEB really improves locality property, traversing down
further to leaf will cause data inside the cache to be replaced more often.

\begin{table}[!t]
\caption{Cache profile comparison during 100\% searching}
\label{tbl:cache}
\centering
\begin{tabular}{| p{1.5in} || p{0.8in} | p{1in} | p{0.5in}| p{0.8in} |}
\hline
\textbf{Tree} 	& \textbf{Load count} & \textbf{Last level cache miss} & \textbf{\%miss} &  \textbf{operations /second}\\
\hline
$\Delta$Tree ($UB=5M$)	&	4301253301	& \multicolumn{1}{|r|}{454239096}	& \multicolumn{1}{|r|}{10.56}	& \multicolumn{1}{|r|}{469942} \\
\hline
$\Delta$Tree ($UB=127$)	&	4895074596	& \multicolumn{1}{|r|}{435682140}	& \multicolumn{1}{|r|}{8.90}	& \multicolumn{1}{|r|}{429945} \\
\hline
SFtree		&	3675847131	& \multicolumn{1}{|r|}{406925489}	& \multicolumn{1}{|r|}{11.07}	& \multicolumn{1}{|r|}{85473} \\
\hline
VTMtree		&	1140447360	& \multicolumn{1}{|r|}{62794247}	& \multicolumn{1}{|r|}{5.51}	& \multicolumn{1}{|r|}{2261378}\\
\hline
\end{tabular}
\end{table}

The bad performance of VTMtree's concurrent update on both of the tree setups
are inevitable, because of the nature of static tree layout. The VTMtree needs to
always maintain a small height, which is done by {\em incrementally}
rebalancing different portions of its structure\cite{BrodalFJ02}. In case of
VTMtree, the whole tree must be locked whenever rebalance is executed, 
blocking other operations. While \cite{BrodalFJ02} explained
that amortised cost for this is small, it will hold true only in when
implemented in the sequential fashion.





\section{Related Work}\label{sec:relatedWork}

The trees involved in the benchmark section are not all the available implementation of the concurrent binary search tree. A novel non-blocking BST was coined in \cite{EllenFRB10}, which subsequently followed by its k-ary implementation \cite{Brown:2011:NKS:2183536.2183551}. These researches are using leaf-oriented tree, the same principle used by $\Delta$Tree and it has a good concurrent operation performance. However the tree doesn't focus on high-performance searches, as the structure used is a normal BST. CBTree \cite{Afek:2012:CPC:2427873.2427875} tried to tackle good concurrent tree with its counting-based self-adjusting feature. But this too, didn't look at how an efficient layout can provide better search and update performance. 
 
Also we have seen the work on concurrent cache-oblivious B-tree\cite{BenderFGK05}, which provides a good overview on how to combine efficient layout with concurrency. However its implementation was far from practical. The recent works in both \cite{Crain:2012:SBS:2145816.2145837,Crain:2013:CBS:2529818.2529849} provides the current state-of-the art for the subject. However none of them targeted a cache-friendly structure which would ultimately lead to a more energy efficient data structure.  


\section{Conclusions and Future Work}\label{sec:conclusions}

We have introduced a new {\em relaxed cache oblivious} model that enables high
parallelism while maintaining the key feature of the original cache oblivious
(CO) model \cite{Prokop99} that analyses for a simple two-level memory are
applicable for an unknown multilevel memory. Unlike the original CO model, the
relaxed CO model assumes a known upper bound on unknown memory block sizes $B$
of a multilevel memory. The relaxed CO model enables developing highly
concurrent algorithms that can utilize fine-grained data locality as desired by
energy efficient computing \cite{Dally11}. 

Based on the relaxed CO model, we have developed a novel {\em dynamic} van Emde
Boas (dynamic vEB) layout that makes the vEB layout suitable for
highly-concurrent data structures with update operations. The dynamic vEB
supports dynamic node allocation via pointers while maintaining the optimal
search cost of  $O(log_B N)$ memory transfers for vEB-based trees of size
$N$ without knowing memory block size $B$. 

Using the dynamic van Emde Boas layout, we have developed $\Delta$Tree that
supports both high concurrency and fine-grained data locality. $\Delta$Tree's
{\em Search} operation is wait-free and its {\em Insert} and {\em Delete}
operations are non-blocking to other {\em Insert, Delete} and {\em Search}
operations. $\Delta$Tree is relaxed cache oblivious: the expected memory
transfer costs of its {\em Search, Delete} and {\em Insert} operations are  
$O(log_B N)$, where $N$ is the tree size and $B$ is unknown memory block size in
the ideal cache model \cite{Frigo:1999:CA:795665.796479}. Our experimental
evaluation comparing $\Delta$Tree with AVL, red-black and speculation-friendly
trees from the the Synchrobench benchmark \cite{synchrobench}
has shown that $\Delta$Tree achieves the best performance when the update
contention is not too high.

\section*{Acknowledgment}

The research leading to these results has received funding from the European Union Seventh Framework Programme (FP7/2007-2013) under grant agreement n$^{\circ}$611183 (EXCESS Project, www.excess-project.eu).

The authors would like to thank Gerth St{\o}lting Brodal for the source code used in \cite{BrodalFJ02}. Vincent Gramoli who provides the source code for Synchrobench from \cite{Crain:2012:SBS:2145816.2145837}. The Department of Information Technology, University of Troms{\o} for giving us access to the Stallo HPC cluster.



%
\bibliographystyle{alpha}
\bibliography{TR-CS-UIT-2013-74}

\newcommand{\etalchar}[1]{$^{#1}$}
\begin{thebibliography}{BFCF{\etalchar{+}}07}

\bibitem[AKK{\etalchar{+}}12]{Afek:2012:CPC:2427873.2427875}
Yehuda Afek, Haim Kaplan, Boris Korenfeld, Adam Morrison, and Robert~E. Tarjan.
\newblock Cbtree: a practical concurrent self-adjusting search tree.
\newblock In {\em Proceedings of the 26th international conference on
  Distributed Computing}, DISC'12, pages 1--15, Berlin, Heidelberg, 2012.
  Springer-Verlag.

\bibitem[AV88]{AggarwalV88}
Alok Aggarwal and S.~Vitter, Jeffrey.
\newblock The input/output complexity of sorting and related problems.
\newblock {\em Commun. ACM}, 31(9):1116--1127, 1988.

\bibitem[BCCO10]{BronsonCCO10}
Nathan~G. Bronson, Jared Casper, Hassan Chafi, and Kunle Olukotun.
\newblock A practical concurrent binary search tree.
\newblock In {\em Proceedings of the 15th ACM SIGPLAN Symposium on Principles
  and Practice of Parallel Programming}, PPoPP '10, pages 257--268, 2010.

\bibitem[BDFC05]{BenderDF05}
Michael Bender, Erik~D. Demaine, and Martin Farach-Colton.
\newblock Cache-oblivious b-trees.
\newblock {\em SIAM Journal on Computing}, 35:341, 2005.

\bibitem[BFCF{\etalchar{+}}07]{BenderFFFKN07}
Michael~A. Bender, Martin Farach-Colton, Jeremy~T. Fineman, Yonatan~R. Fogel,
  Bradley~C. Kuszmaul, and Jelani Nelson.
\newblock Cache-oblivious streaming b-trees.
\newblock In {\em Proceedings of the nineteenth annual ACM symposium on
  Parallel algorithms and architectures}, SPAA '07, pages 81--92, 2007.

\bibitem[BFGK05]{BenderFGK05}
Michael~A. Bender, Jeremy~T. Fineman, Seth Gilbert, and Bradley~C. Kuszmaul.
\newblock Concurrent cache-oblivious b-trees.
\newblock In {\em Proceedings of the seventeenth annual ACM symposium on
  Parallelism in algorithms and architectures}, SPAA '05, pages 228--237, 2005.

\bibitem[BFJ02]{BrodalFJ02}
Gerth~St{\o}lting Brodal, Rolf Fagerberg, and Riko Jacob.
\newblock Cache oblivious search trees via binary trees of small height.
\newblock In {\em Proceedings of the thirteenth annual ACM-SIAM symposium on
  Discrete algorithms}, SODA '02, pages 39--48, 2002.

\bibitem[BH11]{Brown:2011:NKS:2183536.2183551}
Trevor Brown and Joanna Helga.
\newblock Non-blocking k-ary search trees.
\newblock In {\em Proceedings of the 15th international conference on
  Principles of Distributed Systems}, OPODIS'11, pages 207--221, Berlin,
  Heidelberg, 2011. Springer-Verlag.

\bibitem[BP12]{BraginskyP12}
Anastasia Braginsky and Erez Petrank.
\newblock A lock-free b+tree.
\newblock In {\em Proceedinbgs of the 24th ACM symposium on Parallelism in
  algorithms and architectures}, SPAA '12, pages 58--67, 2012.

\bibitem[CGR12]{Crain:2012:SBS:2145816.2145837}
Tyler Crain, Vincent Gramoli, and Michel Raynal.
\newblock A speculation-friendly binary search tree.
\newblock In {\em Proceedings of the 17th ACM SIGPLAN symposium on Principles
  and Practice of Parallel Programming}, PPoPP '12, pages 161--170, New York,
  NY, USA, 2012. ACM.

\bibitem[CGR13]{Crain:2013:CBS:2529818.2529849}
Tyler Crain, Vincent Gramoli, and Michel Raynal.
\newblock A contention-friendly binary search tree.
\newblock In {\em Proceedings of the 19th international conference on Parallel
  Processing}, Euro-Par'13, pages 229--240, Berlin, Heidelberg, 2013.
  Springer-Verlag.

\bibitem[Com79]{Comer79}
Douglas Comer.
\newblock Ubiquitous b-tree.
\newblock {\em ACM Comput. Surv.}, 11(2):121--137, 1979.

\bibitem[CSRL01]{CormenSRL01}
Thomas~H. Cormen, Clifford Stein, Ronald~L. Rivest, and Charles~E. Leiserson.
\newblock {\em Introduction to Algorithms}.
\newblock McGraw-Hill Higher Education, 2nd edition, 2001.

\bibitem[Dal11]{Dally11}
Bill Dally.
\newblock Power and programmability: The challenges of exascale computing.
\newblock In {\em DoE Arch-I presentation}, 2011.

\bibitem[Dre07]{Drepper07whatevery}
Ulrich Drepper.
\newblock What every programmer should know about memory, 2007.

\bibitem[DSS06]{DiceSS2006}
Dave Dice, Ori Shalev, and Nir Shavit.
\newblock Transactional locking ii.
\newblock In {\em Proceedings of the 20th international conference on
  Distributed Computing}, DISC'06, pages 194--208, 2006.

\bibitem[EFRvB10]{EllenFRB10}
Faith Ellen, Panagiota Fatourou, Eric Ruppert, and Franck van Breugel.
\newblock Non-blocking binary search trees.
\newblock In {\em Proceedings of the 29th ACM SIGACT-SIGOPS symposium on
  Principles of distributed computing}, PODC '10, pages 131--140, 2010.

\bibitem[FLPR99]{Frigo:1999:CA:795665.796479}
Matteo Frigo, Charles~E. Leiserson, Harald Prokop, and Sridhar Ramachandran.
\newblock Cache-oblivious algorithms.
\newblock In {\em Proceedings of the 40th Annual Symposium on Foundations of
  Computer Science}, FOCS '99, pages 285--, Washington, DC, USA, 1999. IEEE
  Computer Society.

\bibitem[Gra]{synchrobench}
Vincent Gramoli.
\newblock Synchrobench: A benchmark to compare synchronization techniques for
  multicore programming.
\newblock \url{https://github.com/gramoli/synchrobench}.

\bibitem[Gra10]{Graefe:2010:SBL:1806907.1806908}
Goetz Graefe.
\newblock A survey of b-tree locking techniques.
\newblock {\em ACM Trans. Database Syst.}, 35(3):16:1--16:26, July 2010.

\bibitem[Gra11]{Graefe:2011:MBT:2185841.2185842}
Goetz Graefe.
\newblock Modern b-tree techniques.
\newblock {\em Found. Trends databases}, 3(4):203--402, April 2011.

\bibitem[KCS{\etalchar{+}}10]{KimCSSNKLBD10}
Changkyu Kim, Jatin Chhugani, Nadathur Satish, Eric Sedlar, Anthony~D. Nguyen,
  Tim Kaldewey, Victor~W. Lee, Scott~A. Brandt, and Pradeep Dubey.
\newblock Fast: fast architecture sensitive tree search on modern cpus and
  gpus.
\newblock In {\em Proceedings of the 2010 ACM SIGMOD International Conference
  on Management of data}, SIGMOD '10, pages 339--350, 2010.

\bibitem[NWF06]{valgrind}
N.~Nethercote, R.~Walsh, and J.~Fitzhardinge.
\newblock Building workload characterization tools with valgrind.
\newblock In {\em Workload Characterization, 2006 IEEE International Symposium
  on}, pages 2--2, 2006.

\bibitem[Pro99]{Prokop99}
Harald Prokop.
\newblock Cache-oblivious algorithms.
\newblock Master's thesis, MIT, 1999.

\bibitem[Smi82]{Smith:1982:CM:356887.356892}
Alan~Jay Smith.
\newblock Cache memories.
\newblock {\em ACM Comput. Surv.}, 14(3):473--530, September 1982.

\bibitem[vEB75]{vanEmdeBoas:1975:POF:1382429.1382477}
P.~van Emde~Boas.
\newblock Preserving order in a forest in less than logarithmic time.
\newblock In {\em Proceedings of the 16th Annual Symposium on Foundations of
  Computer Science}, SFCS '75, pages 75--84, Washington, DC, USA, 1975. IEEE
  Computer Society.

\end{thebibliography}

\end{document}